\pdfoutput=1 
\documentclass[11pt]{article}
\usepackage{geometry}
\providecommand{\tabularnewline}{\\}

\pdfoutput=1
\usepackage[utf8]{inputenc}
\usepackage[english]{babel}
\usepackage{booktabs}
\usepackage{csquotes}
\usepackage{amsmath}
\usepackage{amsthm}
\usepackage{amsfonts}
\usepackage{thmtools}
\usepackage{thm-restate}
\usepackage{xcolor}
\usepackage[colorlinks,linkcolor=violet,citecolor=violet,urlcolor=violet,destlabel]{hyperref}
\usepackage{cmbright}
\usepackage[capitalise,noabbrev,nameinlink]{cleveref}

\title{Counting unitaries of T-depth one}
\author{Vadym Kliuchnikov}

\usepackage{graphicx}

\newtheorem{thm}{Theorem}[section]
\newtheorem{lem}[thm]{Lemma}
\newtheorem{prop}[thm]{Proposition}

\newtheorem{dfn}[thm]{Definition}

\usepackage[style=alphabetic]{biblatex}
\begin{document}

\maketitle

\begin{abstract}
We show that the number of T-depth one unitaries on $n$ qubits is
$
\sum_{m=1}^{n}\tfrac{1}{m!}\prod_{k=0}^{m-1}(4^n/2^k - 2^k) \times \# \mathcal{C}_n,
$
where $\#\mathcal{C}_n$ is the size of the $n$-qubit Clifford group, that is the number of unitaries of T-depth zero.
The number of  T-depth one unitaries on $n$ qubits grows as $2^{\Omega(n^2)} \cdot \# \mathcal{C}_n$.
\end{abstract}

\section{Introduction}

Recall that \textbf{$T$-depth one unitaries} (over Clifford+T gate set) on $n$-qubits are the unitaries that can be expressed (up to a global phase) as
\begin{equation}
C_1 ( T^{k_1} \otimes (T^\dagger)^{k_2} \otimes I^{n-k_1-k_2}) C_2  \label{eq:t-depth-one-definition}
\end{equation}
where $C_1$ can $C_2$ are some Clifford unitaries.
Note that ordering of $T$ and $T^\dagger$ in the expression above is not important  because $\mathrm{SWAP}$
gates are Clifford unitaries. 
Because $T^\dagger = TZS$ and $SZ$ is a Clifford, we can always choose $k_1$ = 0 or $k_2 = 0$. 
The goal of this short paper is to calculate the total number of unitaries of $T$-depth one
and establish a canonical form for Clifford unitaries of $T$-depth $d$.
These results can be used as a part of various counting arguments. For example, one may use our results to lower-bound $T$ depth required to implement an $n$-qubit reversible function in the worst case.

\section{Main result}

We prove the following result
\begin{thm}[T-depth one count] \label{thm:t-depth-one-count}
The number of T-depth one unitaries on $n$ qubits is
\begin{equation*}
\sum_{m=1}^{n}\tfrac{1}{m!}\prod_{k=0}^{m-1}(4^n/2^k - 2^k)\times \# \mathcal{C}_n, \text{ where }
\end{equation*}
$\#\mathcal{C}_n$ is the size of the Clifford group, that is the number of unitaries of T-depth zero.
\end{thm}
In particular, the number of equivalence classes up to right multiplication by
a Clifford unitary is
$$
 \sum_{k=1}^m \tfrac{1}{m!} \prod_{k=0}^{m-1}(4^n/2^k - 2^k)
$$
which grows as $2^{\Omega(n^2)}$.
To prove the theorem, we explicitly count the number of $T$-depth
one unitary matrices that require $m$ $T$ gates.

Let us first introduce a canonical form for unitaries of T-depth one, expanding on \cite{Gosset2014}. 
Recall that a set of Pauli operators $\{ P_1,\ldots,P_m \}$ is \textbf{independent} when none of $P_j$ is equal to the product of a subset of $\{ P_1,\ldots,P_{j-1}$, $P_{j+1},\ldots,P_{m} \}$ up to a sign.
\begin{prop} \label{prop:canonical-form}
Any $n$-qubit unitary of T-depth one can be written (up to a global phase) as: 
\begin{equation}
\exp( i \pi P_1 /8) \ldots \exp( i \pi P_m /8) C \label{eq:canonical}
\end{equation}
where $P_1,\ldots,P_m$ are  commuting independent Pauli operators from $\{I,X,Y,Z\}^{\otimes n}$ and $C$ is an $n$-qubit Clifford unitary. Every unitary given by \cref{eq:canonical}
is a T-depth one unitary and requires at most $m$ T gates.
\end{prop}
\begin{proof}
Recall that $T^\dagger = \exp(i \pi Z/8)$ up to a global phase
and that $C \exp(i \pi Z_1/8) C^\dagger = \exp(i \pi P/8)$,
where $P = C Z_1 C^\dagger$ is Pauli operator, in other words
$P$ is the result of conjugating $Z_1$ by Clifford $C$.
By rewriting \cref{eq:t-depth-one-definition}
$$
C_1 ( (T^\dagger)^{m} \otimes I^{n-m}) C_2 =
C_1\left(( T^\dagger)^{m} \otimes I^{n-m}\right) C^\dagger_1 C_1 C_2 = \exp( i \pi P_1 /8) \ldots \exp( i \pi P_m /8) C_1 C_2
$$
we see that any such unitary can be written as product of $m$
exponents $\exp(i \pi P_j /8)$ where
$\{ P_1, \ldots, P_m \}$ is a set of commuting independent Pauli operators, that is each $P_j$ is from $\pm\{I,X,Y,Z\}^{\otimes n}$.
The commutation and independence follow because the set
$\{ P_1, \ldots, P_k \}$  is obtained by conjugating another set of commuting independent Pauli operators $Z_1,\ldots, Z_m$ by Clifford $C$.
Note that it is well-known that for every such set of Pauli operators there exist a Clifford unitary $C_3$ such that $C_3 P_j C_3^\dagger =-Z_j$, where
$Z_j$ is the $n$-qubit Pauli matrix with $Z$ on qubit $j$ and identity
on the rest of the qubits.
For this reason, any unitary expressed as \cref{eq:canonical}
is a T-depth one unitary with $m$ T gates when $P_j$ are commuting independent Pauli operators and $C$  is an arbitrary Clifford unitary. Moreover, we can choose 
all $P_j$ to be from the set of Pauli matrices $\{I,X,Y,Z\}^{\otimes n}$ (that is 
always with $+$ in front of them). This is because $\exp( - i \pi  P/8) = \exp( - i \pi  P/4) \exp(i \pi  P/8)$ and $\exp( - i \pi  P/4)$ is a Clifford unitary.
\end{proof}

\cref{thm:t-depth-one-count} is a corollary of the following two results we will prove later.
\begin{lem}[Distinctness] \label{lem:distinctness}
Let $n \ge 1$ and $\mathcal{P} = \{ P_1, \ldots, P_m\}$,$\mathcal{Q} = \{ Q_1, \ldots, Q_m\}$
be two sets of independent commuting $n$-qubit Pauli operators and $C_1,C_2$ be two 
$n$-qubit Clifford unitaries. Then unitaries
\begin{equation} \label{eq:canonical-equality}
\exp( i \pi P_1 /8) \ldots \exp( i \pi P_m /8) C_1 = \exp( i \pi Q_1 /8) \ldots \exp( i \pi Q_m /8) C_2 \text{ (up to a global phase)}
\end{equation}
if and only if $\mathcal{P} = \mathcal{Q}$ as sets and $C_1 = C_2$ up to a global phase.
\end{lem}

\begin{restatable}[T-count]{lem}{tcount} 
\label{lem:t-count}
Any unitary
\begin{equation*}
\exp( i \pi P_1 /8) \ldots \exp( i \pi P_m /8) C 
\end{equation*}
where $P_1,\ldots,P_m$ are  commuting independent Pauli operators from $\pm \{I,X,Y,Z\}^{\otimes n}$ and $C$ is an $n$-qubit Clifford unitary
requires exactly 
$m$ T gates.
\end{restatable}

\begin{proof}[Proof of~\cref{thm:t-depth-one-count}]
Let $N_{m,n}$ be the number of $T$ depth one unitaries on $n$ qubits that require $m$ $T$ gates. The total number of $T$ depth one unitaries is $\sum_{m=1}^n N_{m,n}$. 
It remains to derive expression for $N_{m,n}$. 
According to \cref{lem:t-count} and \cref{prop:canonical-form}
every $T$-depth one unitary with $m$ T gates can be expressed by \cref{eq:canonical}. 
According to \cref{lem:distinctness}, distinct sets of independent commuting Pauli operators $\{ P_1, \ldots, P_m \}$ and Clifford unitaries $C$ correspond to distinct unitaries in \cref{eq:canonical}. For this reason 
$$
N_{m,n} = \prod_{k=0}^{m-1}(4^n/2^k - 2^k) / m! \cdot \# \mathcal{C}_n
$$
To derive above expression we used the fact that there are  $\prod_{k=0}^{m-1}(4^n/2^k - 2^k)$ $m$-tuples of commuting independent Pauli 
operators (without signs). For more details see  the proof of \href{https://arxiv.org/pdf/quant-ph/0406196.pdf\#page=3}{Proposition~2} in \cite{Aaronson2004}. We divide by $m!$ to account for possible permutations 
of tuples because we need to count the distinct sets. Finally, we multiply by 
$\# \mathcal{C}_n$ to account for all possible $C$ in \cref{eq:canonical}.
\end{proof}

In our proofs we rely on \textbf{the channel representation} $\hat{U}$~\cite{Gosset2014} of a $n$-qubit unitary $U$.
It is $4^n\times 4^n$ real matrix with rows and columns indexed by $n$-qubit Pauli matrices $\{I,X,Y,Z\}^{\otimes n}$,
where the entry of $\hat{U}$ is defined as
$$
\hat{U}_{P,Q} = \frac{1}{2^n} \mathrm{Tr}(PUQU^\dagger)
$$
It has been shown that the channel representation of any
unitary matrix is a real orthogonal matrix and channel representation
of a Clifford matrix is a  signed permutation matrix(product of a permutation matrix and a diagonal matrix with $\pm 1$ on the diagonal).
We start with proving a special case of \cref{lem:distinctness}.

For bit-string $a \in \{0,1\}^n$ of length $n$ let us introduce notation
\begin{equation} \label{eq:bit-string-power}
Z^a = Z^{a(1)}\otimes \ldots \otimes Z^{a(n)}, \text{ where } a = a(1),a(2),\ldots,a(n)
\end{equation}
We call Pauli operators $Z^a$ \textbf{positive diagonal Pauli operators}.

\begin{lem}[Diagonal equality] \label{lem:diagonal-equality} 
Let $P_1,\ldots,P_n$ be positive diagonal independent commuting $n$-qubit Pauli operators
and let $C$ be an $n$-qubit Clifford matrix,
then equality
$$
\exp( i \pi P_1 /8) \ldots \exp( i \pi P_n /8) = \exp( i \pi Z_1 /8) \ldots \exp( i \pi Z_n /8) C
$$
up to a global phase implies that set $\{ P_1, \ldots, P_n \}$ is equal to set
$\{ Z_1, \ldots , Z_n \}$
and $C$ is the identity up a global phase.  
\end{lem}

In the proof, we will use the following lemmas that we prove in \cref{sec:helpful-lemmas}.

\begin{restatable}[Diagonal Clifford image]{lem}{diagonalcliffordimage} 
\label{lem:diagonal-clifford-image}
Let $C$ be a diagonal $n$-qubit Clifford unitary, then
$C X_j C^\dagger = i^k X_j Z^a$ for some integer $k$
and bit-string $a \in \{0,1\}^n$.
\end{restatable}

\begin{restatable}[Hamming weight]{lem}{hammingweight} 
\label{lem:hamming-weight}
Let $a_1,\ldots,a_k$ be a set of linear independent bit-strings (as vectors over $\mathbb{F}_2$) and let
$$
M = \exp( i \pi Z^{a_1} /8) \ldots \exp( i \pi Z^{a_k} /8)  X_j  \exp( -i \pi Z^{a_1} /8) \ldots \exp( -i \pi Z^{a_k} /8)
$$
Consider $m_P$ to be coefficients of expanding  $M$ in a Pauli basis as defined below
\begin{equation} \label{eq:pauli-basis}
 M = \sum_{P \in \{I,X,Y,Z\}^{\otimes n}} m_P P, 
\end{equation}
Then the Hamming weight of $m_P$ is equal to $2^w$,
where $w$ is the Hamming weight of $a_1(j),\ldots,a_k(j)$.
\end{restatable}

We proceed to prove \cref{lem:diagonal-equality}.

\begin{proof}[Proof of~\cref{lem:diagonal-equality}]
Consider image $M$ of $X_j$ under
$$
\exp( i \pi Z_1 /8) \ldots \exp( i \pi Z_n /8) C
$$
Note that Clifford $C$ must be diagonal because it can be expressed as a product of diagonal matrices.
Because Clifford $C$ is diagonal, according to \cref{lem:diagonal-clifford-image} $M$ is equal to the image of $i^{k'}X_j Z^{a'}$
under conjugation by
$$
U = \exp( i \pi Z_1 /8) \ldots \exp( i \pi Z_n /8)
$$
for some bit-string $a'$ and integer $k'$. That is $M$ is equal
to $i^{k'} Z^{a'} UX_jU^\dagger$.

Define $m_P$ to be the expansion of
$M$ in Pauli basis as in \cref{eq:pauli-basis}.
The Hamming weight of $m_P$ is equal to two.
This is because Hamming weight of the expansion in Pauli basis of $UX_jU^\dagger$ is two according to \cref{lem:hamming-weight},
and the fact that hamming weight of expansions of $M$ and $UX_jU^\dagger= M (-i)^{k'} Z^{a'}$ are the same.

Let us represent $P_k$  as $Z^{a_k}$ for bit-strings $a_k$.
Pauli operators $Z^{a_k}$ are independent if and only if bit-strings $a_k$
are linearly independent as vectors over $\mathbb{F}_2$.
Now, using \cref{lem:hamming-weight} again
we see that the Hamming weight of bit-strings
$a_1(j), \ldots, a_n(j)$ must be one for all $j$.
Because all $P_k$ are independent, the only possible
way for this to happen is if set $\{ P_1,\ldots, P_n\}$
is equal to the set $\{ Z_1,\ldots, Z_n\}$ and Clifford $C$ is
identity up to a global phase. This completes the proof.
\end{proof}

The proof of \cref{lem:distinctness} relies on the following lemma shown in \cref{sec:helpful-lemmas}.
\begin{restatable}[Unit rows]{lem}{unitrows} 
\label{lem:unit-rows}
Let $\{ P_1, \ldots, P_m\}$ be a set of commuting independent $n$-qubit Pauli operators 
and let $C$ be $n$-qubit Clifford unitary.
In channel representation of 
$$
\exp( i \pi P_1 /8) \ldots \exp( i \pi P_m /8) C
$$
the only rows $\pm 1$ are the ones indexed by
$$
 \left\{ P: 
 P \in \{I,X,Y,Z\}^{\otimes n} \text{ and } P \text{ commutes with } P_1,\ldots,P_m \right\}
$$
\end{restatable}

\begin{proof}[Proof of \cref{lem:distinctness}]
Our goal is to reduce the proof to \cref{lem:diagonal-equality}.
Clearly, with $C_3 = C_1 C_2^\dagger$ we have
$$
\exp( i \pi P_1 /8) \ldots \exp( i \pi P_n /8) C_3 = \exp( i \pi Q_1 /8) \ldots \exp( i \pi Q_n /8) \text{ (up to global phase)}.
$$
Let now $C_4$ be a Clifford such that $C_4 Q_k C_4^\dagger = Z_k$, for $k=1,\ldots,m$,
and let $P'_k = C_4 P_k C_4^\dagger$. We conjugate equation above by $C_4$ and get
$$
\exp( i \pi P'_1 /8) \ldots \exp( i \pi P'_m /8) C_4 C_3 C_4^\dagger = \exp( i \pi Z_1 /8) \ldots \exp( i \pi Z_m /8)  \text{ (up to global phase)}
$$
Now introducing $C^\dagger = C_4 C^\dagger_3 C_4^\dagger$ we have:
\begin{equation} \label{eq:diag-equality-2}
   \exp( i \pi P'_1 /8) \ldots \exp( i \pi P'_m /8) = \exp( i \pi Z_1 /8) \ldots \exp( i \pi Z_m /8) C  \text{ (up to global phase)} 
\end{equation}
It remains to show that $P'_k$ are diagonal Pauli operators supported on first $m$ qubits.
Above equality up to a global phase implies that 
the channel representation of the right-hand side and left-hand side
must be the same. In particular, they have the same rows that contain $\pm 1$, 
and therefore, according to \cref{lem:unit-rows}, the following sets are equal: 
$$
 \left\{ \text{Pauli matrices that commute with }  P'_1,\ldots,P'_m  \right\} =  \left\{  \text{Pauli matrices that commutes with } Z_1, \ldots, Z_m \right\}
$$
The above set is equal to
$$
\langle Z_1, \ldots, Z_m \rangle \otimes \{I,X,Y,Z\}^{\otimes (n-m)} 
$$
For this reason, each $P'_k$ is a diagonal Pauli operator supported on first $m$ qubits.
We can choose $P'_k$ to be diagonal positive Pauli operators, 
because $\exp(i \pi P'_k / 8) = \exp(-i \pi P'_k / 8) \exp(i \pi P'_k / 4)$ 
and $\exp(i \pi P'_k / 4)$ is a Clifford.
Clifford $C$ must be of the from $A \otimes I_{2^{n-m}}$\footnote{$I_d$ is a $d$-dimensional identity matrix} because 
it can be expressed as a product of two matrices of this form according to \cref{eq:diag-equality-2}.
Applying \cref{lem:diagonal-equality}
completes the proof.
\end{proof}

\section{Some properties of images of Pauli operators and the channel representation} \label{sec:helpful-lemmas}

\diagonalcliffordimage*
\begin{proof}
Because $C$ is diagonal, for all $k$ we have $C Z_k C^\dagger = Z^k$.
For $k\ne j$ Pauli matrices $Z_k$ and $X_j$ commute, therefore image
$C X_j C^\dagger$ commutes with image $C Z_k C^\dagger = Z_k$ too.
Denote $C X_j C^\dagger = P_1\otimes \ldots \otimes P_k$.
The commutativity constraint implies that $P_k \in \{I,Z\}$
for $k \ne j$, which shows required result.
\end{proof}

\hammingweight*
\begin{proof}
Let us first look at the expression for:
$$
\exp( i \pi Z^a /8) X_j \exp( -i \pi Z^a /8)
$$
It is equal to $X_j$ if bit $a(j)$ is zero,
when the bit is one the expression becomes
$$
\frac{1}{\sqrt2}X_j(I + i Z^a)
$$
By repeatedly applying the above observation, we see that
$$
M = X_j \prod_{l : a_l(j) =1} \frac{1}{\sqrt2}(I + i Z^{a_l})
$$
Because all bit-strings $a_l$ are linearly independent, we see that product
$$
\prod_{l : a_l(j) =1} (I + i Z^{a_l})
$$
is equal to the sum of $2^w$ distinct Pauli operators.
\end{proof}

\unitrows*
\begin{proof}
First consider the case when $P_1,\ldots,P_m$ is equal to $Z_1,\ldots,Z_m$
and $C$ is identity. The channel representation of $\exp( i \pi Z_1 /8) \ldots \exp( i \pi Z_m /8)$ is equal to $R^{\otimes m} \otimes I_{4^{m-n}}$
, where $R$ is the channel representation of $\exp(i\pi Z/8)$. See \href{https://arxiv.org/pdf/1308.4134.pdf#page=7}{Equation (4.2)} in~\cite{Gosset2014} for the expression for R. In this case columns and rows indexed by 
$$
\langle Z_1,\ldots,Z_m\rangle \otimes \{I,X,Y,Z\}^{\otimes (n-m)}
$$
contain $\pm 1$ and the rest of the columns and rows do not contain $\pm 1$.
These are exactly the columns indexed by Pauli matrices that commute and anti-commute with $Z_1,\ldots, Z_m$. Let now $C_1$ be a Clifford that maps $Z_1,\ldots,Z_m$ to 
$P_1,\ldots,P_m$ by conjugation:
$$
C_1 \exp( i \pi Z_1 /8) \ldots \exp( i \pi Z_m /8) C^\dagger_1 = \exp( i \pi P_1 /8) \ldots \exp( i \pi P_m /8).
$$
Conjugation by $C_1$ preserves commutativity and anti-commutativity, it also simultaneously permutes rows and columns of the channel representation and flips signs. For this reason, the only rows of channel representation of $\exp( i \pi P_1 /8) \ldots \exp( i \pi P_m /8)$ with $\pm1$ are rows indexed by Pauli matrices that commute with $P_1,\ldots,P_m$. Right-multiplication by Clifford $C$ permutes columns and flips signs, 
therefore rows that contain $\pm 1$ stay the same.
\end{proof}

\tcount*
\begin{proof}
According to \cref{prop:canonical-form}, any unitary $U$ given by \cref{eq:canonical}
requires at most $m$ T gates.
It remains to show that such unitaries require at least $m$ T gates.
It is sufficient to show that at least $m$ T states are required to prepare Choi state
$$
|U\rangle = \frac{1}{\sqrt{2}^n} (I_{2^n} \otimes U) \sum_{k \in \{0,1\}^n} |k\rangle \otimes |k\rangle 
$$
We will use lower-bound on the number of $T$ states needed to prepare a state from \cite{Beverland2020} in terms of dyadic monotone $\mu_2$; see
\href{https://arxiv.org/pdf/1904.01124.pdf#page=27}{Definition 6.2} in \cite{Beverland2020}.
The dyadic monotone is computed using the Pauli spectrum 
$$
 \mathrm{Spec}|\psi\rangle = \left\{ |\langle \psi | P | \psi \rangle|  : P \in \{I,X,Y,Z\}^{\otimes n} \right\}
$$
Recall that in \href{https://arxiv.org/pdf/1904.01124.pdf#page=60}{Appendix 10.1} in~ \cite{Beverland2020} it has been shown that Pauli spectrum of the Choi state $|U\rangle$ is exactly the set of absolute values of entries of the channel representation of $U$.
For this reason Pauli spectrum of $U$ is the same as Pauli spectrum of $T^{\otimes m}\otimes I_{2^{n-m}}$. For this reason $\mu_2 |U\rangle = m/2$ 
and there at least $m$ T states needed to prepare $|U\rangle$.
\end{proof}

\section{Concluding remarks}

A corollary of \cref{prop:canonical-form} and \cref{lem:distinctness} is a canonical form for unitaries 
with $T$ depth $d$. Our canonical form is inspired by a canonical form introduced in \cite{Mosca2021}. Lemma \cref{lem:distinctness} motivates the following definition.

\begin{dfn}
\label{dfn:generators}
Let $n$ be a positive integer and let $m \le n$ be another positive integer, 
define sets 
$$
G_{n,m} = \left\{ e^{ i \pi P_1 /8} \ldots e^{ i \pi P_m /8} : P_1,\ldots,P_m \text{ independent and commuting elements of } \{I,X,Y,Z\}^{\otimes n} \right\}.
$$
Define $G_n = \bigcup_{m = 1}^n G_{n,m}$.
\end{dfn}

\cref{lem:distinctness} shows that all elements of  $G_n$  are distinct up to the right multiplication by a Clifford 
unitary. We have numerically verified this fact on up to four qubits. Sets $G_{n,m}$ are exactly the sub-sets of $G_n$
that contain unitaries of $T$-count $m$. We provide sizes of sets $G_n$ and $G_{n,m}$ in \cref{tab:t-depth-one-sizes}.
Note that the number of unitaries of $T$-depth one is $\# G_n \cdot \#\mathcal{C}_{n}$
and the number of unitaries of $T$-depth one and $T$-count $m$ is $\# G_{n,m} \cdot \#\mathcal{C}_{n}$.
We introduce a canonical form in the following theorem:

\begin{table}[]
    \centering
    { \scriptsize
\begin{tabular}{|c|c|c|c|c|c|c|}
\hline 
Number & Number of unitaries & Number of  & \multicolumn{4}{c|}{Number of unitaries of $T$-depth one and given $T$-count}\tabularnewline
\cline{4-7} \cline{5-7} \cline{6-7} \cline{7-7} 
of qubits & of $T$-depth one & Clifford unitaries $\#\mathcal{C}_{n}$ & $T$-count 1 & $T$-count 2 & $T$-count  3 & $T$-count 4\tabularnewline
\hline 
1 & $3\cdot\#\mathcal{C}_{n}$ & $24$ & $3\cdot\#\mathcal{C}_{n}$ & - & - & -\tabularnewline
\hline 
2 & $60\cdot\#\mathcal{C}_{n}$ & $11520$ & $15\cdot\#\mathcal{C}_{n}$ & $45\cdot\#\mathcal{C}_{n}$ & - & -\tabularnewline
\hline 
3 & $4788\cdot\#\mathcal{C}_{n}$ & $92897280$ & $63\cdot\#\mathcal{C}_{n}$ & $945\cdot\#\mathcal{C}_{n}$ & $3780\cdot\#\mathcal{C}_{n}$ & -\tabularnewline
\hline 
4 & $2265420\cdot\#\mathcal{C}_{n}$ & $12128668876800$ & $255\cdot\#\mathcal{C}_{n}$ & $16065\cdot\#\mathcal{C}_{n}$ & $321300\cdot\#\mathcal{C}_{n}$ & $1927800\cdot\#\mathcal{C}_{n}$\tabularnewline
\hline 
\end{tabular}
}
    \caption{Number of unitaries of $T$-depth one on one to four qubits. Results in this table has been verified numerically. }
    \label{tab:t-depth-one-sizes}
\end{table}

\begin{thm}
Let $U$ be an $n$-qubit unitary of $T$-depth $d$, then it can be written as a product
$
 U_1 \ldots U_d C
$  (up to a global phase),
where $U_k$ are from $G_n$~(\cref{dfn:generators}) and $C$ is an $n$-qubit Clifford unitary.
\end{thm}
\begin{proof}
The proof is similar to \cref{prop:canonical-form} and is a slight generalization of the proof
of a canonical form in \cite{Gosset2014}.
\end{proof}

Interestingly, the sets $G_n$ we use for our canonical form has an optimal size. 
More precisely, the following proposition is true:

\begin{prop}
Suppose there exist a family of sets $\tilde{G}_n$ such that any $n$-qubit unitary of $T$-depth $d$ can be written as a product
$
 U_1 \ldots U_d C
$  (up to a global phase),
where $U_k$ are from $\tilde{G}_n$~(\cref{dfn:generators}) and $C$ is an $n$-qubit Clifford unitary, 
then $\# \tilde{G}_n \ge \# G_n$.
\end{prop}
\begin{proof}
For any family of sets  $\tilde{G}_n$, the number of $T$-depth one unitaries must be upper-bounded by $\# \tilde{G}_n \cdot \# \mathcal{C}_n$, therefore $\# \tilde{G}_n \cdot \# \mathcal{C} \ge \# G_n \cdot \# \mathcal{C}_n$.
\end{proof}

\printbibliography

\end{document}